\documentclass[epsf]{iopart}
\expandafter\let\csname equation*\endcsname\relax
\expandafter\let\csname endequation*\endcsname\relax
\usepackage{iopams,amsfonts,amsthm,rotating,graphicx,fancybox,fancyhdr,bm,eucal,dsfont,appendix,mathtools}
\usepackage{datetime,esint,color}
\usepackage{accents}
\usepackage{wasysym}
\usepackage{cite}
\usepackage[normalem]{ulem}
\usepackage{booktabs, siunitx}
\usepackage{caption}
\usepackage{rotating}
\usepackage{hyperref}
\definecolor{JMPLinkBlue}{RGB}{48,169,224}
\hypersetup{colorlinks=true, urlcolor=JMPLinkBlue, linkcolor=black, citecolor=black}

\eqnobysec
\settimeformat{hhmmsstime}
\ddmmyyyydate

\theoremstyle{plain}
\newtheorem{theorem}{Theorem}[section]
\newtheorem{lemma}[theorem]{Lemma}
\newtheorem{proposition}[theorem]{Proposition}

\theoremstyle{definition}

\newtheorem{remark}[theorem]{Remark}

\newtheorem*{conjectureA}{Conjecture A}
\newtheorem*{conjectureB}{Conjecture B}
\newtheorem*{conjectureC}{Conjecture C}

\makeatletter
\renewcommand\subsection{\@startsection{subsection}{2}{\z@}%
  {-18pt plus -4pt minus -4pt}%
  {8pt plus 2pt minus 2pt}%
  {\normalfont\normalsize\bfseries}}
\makeatother

\begin{document}

\title{On the asymptotic duality of spectral variances in random matrix theory and the ``$1/6$'' formula}

\author{\hspace{-2cm}}

\author{\hspace{-2cm}Peng Tian$^{1}$, Roman Riser$^{2,3}$ and Eugene Kanzieper$^3$\,\footnote[1]{Corresponding author. Email: eugene.kanzieper@hit.ac.il}}

\address{
\hspace{-2cm}$^1$~MAP5, Universit\'e Paris Cit\'e, Paris 75006, France\\
\hspace{-2cm}$^{2}$~Department of Mathematics, Tulane University, New Orleans, LA 70118, United States\\
\hspace{-2cm}$^3$~School of Mathematical Sciences, Holon Institute of Technology, Holon 5810201, Israel\\
}
\par\vspace{18pt}
{\centering\normalfont\itshape
Dedicated to Uzy Smilansky on the occasion of his 85th birthday.\par}
\vspace{6pt}
\begin{abstract}
A ``mysterious'' relation between the number variance and the variance of the $L$-th ordered eigenvalue, first suggested by French et al.~[\href{https://doi.org/10.1016/0003-4916(78)90205-1}{Ann.~Phys.}~{\bf 113}, 277 (1978)], is revisited and proven to be asymptotically exact for the $\beta=2$ Dyson symmetry class. Central to the proof is a previously unknown sum rule for the level spacing auto-covariances. Its derivation hinges on our previous work on the power spectrum description of eigenvalue fluctuations in random matrix theory. Analytical results for $\beta=2$ are complemented by conjectural extensions to the $\beta=1$ and $\beta=4$ symmetry classes. Our findings are corroborated by a comprehensive numerical analysis.
\par\vspace{5\baselineskip}
\noindent Published in:~\href{https://doi.org/10.1063/5.0340386}{J.~Math.~Phys. {\bf 67}, 123501 (2026)}.
\end{abstract}
\newpage
\section{Introduction}\label{I-section}

A vast amount of experimental, numerical, and theoretical
evidence~\cite{GMGW-1998,S-1999,H-2001} suggests that the spectra of
`maximally chaotic' quantum systems exhibit a remarkable statistical
universality~\cite{B-1987}, stemming from underlying symmetries rather than
from the microscopic peculiarities of individual systems. These universal
fluctuation laws are accurately described by (infinite-dimensional) random
matrix theory~\cite{M-2004,PF-book} at both short- and long-range energy
scales.

To isolate these universal fluctuations, the system-specific mean level
density must be eliminated from both empirical data and theory. In the
latter case, this is achieved through the unfolding procedure~\cite{M-2004}
\begin{eqnarray} \label{unfolding}
    \lambda_\ell = \int_{-\infty}^{\varepsilon_\ell} d\epsilon \,\varrho_N(\epsilon), \quad \ell=1,\dots,N,
\end{eqnarray}
that maps the spectrum $\{ \varepsilon_1,\dots,\varepsilon_N\}$ of a
finite-$N$ parent random matrix model onto a set $\{
\lambda_1,\dots,\lambda_N\}$ of unfolded eigenlevels, using
$\varrho_N(\epsilon)$, its mean level density. In the limit $N \rightarrow
\infty$, the fluctuations in the
bulk of the unfolded spectrum -- described by the ${\rm Sine}_\beta$ point
process~\cite{KS-2009,VV-2017,MNN-2019} -- are then
expected~\cite{BGS-1984} to obey universal statistical laws
observed in fully chaotic quantum systems. If the mean level
spacing in the unfolded spectrum is set to
unity, these universal laws depend on a single parameter --
the Dyson symmetry index $\beta$, which we restrict here to the
classical values $1$, $2$ and $4$,
corresponding, respectively, to the orthogonal, unitary, and
symplectic symmetry classes in the Dyson-Mehta
terminology~\cite{M-2004}.

Since the inception of random matrix theory, a variety of statistical
measures have been devised to accurately study fluctuations in the bulk
spectra of bounded quantum systems. While the traditional classification
distinguishes between short- and long-range statistical indicators
(highlighting spectral correlations on the local and global energy scales,
respectively), we follow Ref.~\cite{RTK-2023} and assign them to two
alternative classes of {\it ordinary} and {\it ordered} linear level
statistics, depending on whether the ordering of sampled eigenvalues is
essential for the definition of the corresponding random variable. The two
classes clearly differ from each other on a formal level due to the much
different mathematical structures at their core.

(i) The ordinary linear spectral statistics deal with the fluctuation
properties of a random variable
\begin{eqnarray} \label{ordinary-rv}
    X_N({\bm \lambda}) = \sum_{\ell=1}^{N} f_N(\lambda_\ell),
\end{eqnarray}
where $f_N(\lambda)$ is a test function of interest and ${\bm \lambda} =
\{\lambda_1,\dots,\lambda_N\}$ are (possibly unfolded) eigenvalues of an
$N\times N$ random matrix. Since $X_N({\bm \lambda})$ treats each eigenvalue
equally, it characterizes the energy spectrum as a whole and primarily
probes its global texture. Being invariant under any permutation of
individual eigenvalues, the ordinary level statistics possess a high degree
of analytical tractability --- the calculation of any finite moment is a
{\it`finite-complexity'} task as it requires knowledge of spectral
correlation functions only of finite order.

Traditional examples that have predominated in the literature
include~\cite{M-2004} two-point (and higher-order) correlation functions,
the spectral form factor, the number variance -- defined as the variance of
the number ${\mathcal N}(s)$ of levels in an interval of length $s$ -- along
with its higher cumulants, and the Dyson--Mehta $\Delta_3$ statistic. In the context of this paper, the ordinary spectral statistics of interest is represented 
by the counting function ${\mathcal N}(s)$ obtained from Eq.~(\ref{ordinary-rv}) by choosing $f_N(\lambda)={\mathds 1}_{I_s}(\lambda)$, where
${\mathds 1}_{I_s}$ is the indicator of a fixed interval $I_s$ of length $s$ in the unfolded bulk, and passing to the infinite-process limit. Clearly, evaluating 
the number variance ${\rm var}_\beta [{\mathcal N}(s)]$ requires only the one- and two-point spectral correlation functions. For future reference, we quote non-perturbative formulas for the number variance ${\rm var}_\beta[{\mathcal N}(s)]$, as given in Eqs.~(5.12)--(5.13) of
Ref.~\cite{BFFMPW-1981}; 
equivalent integral representations are specified by Eqs.~(16.1.2)--(16.1.4) of Ref.~\cite{M-2004}:~\cite{Note-Brody-Misprint}
\begin{eqnarray} \label{var-2} \fl\qquad
    {\rm var}_2 [{\mathcal N}(s)] = \frac{1}{\pi^2}\Big( 1+\gamma+\log(2\pi s)\Big) \nonumber\\
    \fl\qquad\qquad\qquad
    - \frac{1}{\pi^2}
    \left\{
        \cos(2\pi s) + {\rm Ci}(2\pi s) + 2\pi s \left(
            {\rm Si}(2\pi s) - \frac{\pi}{2}
        \right)
    \right\},
\end{eqnarray}
\begin{eqnarray} \label{var-1} \fl\qquad
    {\rm var}_1 [{\mathcal N}(s)] = 2 \,{\rm var}_2 [{\mathcal N}(s)] + \frac{{\rm Si}(\pi s)}{\pi} \left(
    \frac{{\rm Si}(\pi s)}{\pi}  - 1
    \right),
\end{eqnarray}
and
\begin{eqnarray} \label{var-4} \fl\qquad
    {\rm var}_4 [{\mathcal N}(s)] = \frac{1}{2} \,{\rm var}_2 [{\mathcal N}(2s)] + \left(\frac{{\rm Si}(2\pi s)}{2\pi} \right)^2.
\end{eqnarray}
Here ${\rm Ci}(x)$ and ${\rm Si}(x)$ are the cosine and sine integrals defined, respectively, by~\cite{DLMF}
\begin{equation}
{\rm Ci}(x) = - \int_{x}^{\infty} \frac{\cos t}{t}dt, \qquad {\rm Si}(x) =  \int_{0}^{x} \frac{\sin t}{t}dt.
\end{equation}
Equations~(\ref{var-2})--(\ref{var-4}) refer to the unfolded spectrum of infinite-dimensional random matrices associated with the ${\rm Sine}_\beta$ point process~\cite{KS-2009,VV-2017,MNN-2019} possessing the unit mean local density, $\rho=1$.

(ii) In stark contrast, by targeting individual eigenlevels, the ordered
level statistics inherently presuppose a sorted spectrum. They can be
represented by a random variable
\begin{eqnarray}\label{ord-stat}
X_N({\bm \lambda};{\bm c}) = \sum_{\ell=1}^N c_\ell f_N(\lambda_\ell),
\end{eqnarray}
where the eigenvalues are ordered, $\lambda_1 \le \dots \le \lambda_N$, and
the sequence of weights ${{\bm c}=\{c_1,\dots,c_N\}}$ is not constant,
${\bm c}\neq c {\mathds 1}_N$. The latter is a game-changer, as it generically causes the finite moments of such statistics to depend on spectral
correlation functions of {\it all} orders
up to $N$. In the bulk limit $N\rightarrow\infty$, this
turns the calculation of their finite moments into an {\it
`$\infty$-complexity'} task, comparable in difficulty to
deriving the full distributions of ordinary statistics. As a consequence,
quantifying even basic fluctuations of ordered level statistics typically
requires the far more sophisticated machinery~\cite{AvM-2001} rooted in the
theory of nonlinear integrable lattices and Painlev\'e transcendents, partly
explaining why this remains relatively under-explored territory.

Consecutive level spacings~\cite{M-2004,PF-book} are a paradigmatic example
of ordered statistics. Proposed at the dawn of random matrix
theory~\cite{W-1956} as a very intuitive indicator, the level spacing
distribution~\cite{G-1961} has become but one of an entire toolbox of
ordered level statistics designed to provide complementary insights into
spectral fluctuations. Other examples include the auto-covariances of level
spacings~\cite{BLS-2001,RTK-2023}, the closely related statistics of local
level spacings~\cite{TRK-2024}, the power spectrum of
spacings~\cite{O-1987,RTK-2023,TRK-2024} and eigenlevel
sequences~\cite{RGMRF-2002,FGMMRR-2004,ROK-2017,ROK-2020,RK-2021,RK-2023,FW-2025},
and $r$-statistics~\cite{OH-2007,ABGR-2013}.

The ordered statistic of particular interest in this paper is
the $L$-th eigenvalue $\lambda_L$. For a finite ordered spectrum,
it is selected in Eq.~(\ref{ord-stat}) by choosing
$f_N(\lambda)=\lambda$ and $c_\ell=\delta_{\ell,L}$.
In the infinite bulk spectrum, we label the
levels relative to a typical reference level, shifted to
$\lambda_0=0$, with $\lambda_L$ denoting the
distance from this reference level to its $L$-th neighbor
to the right. Unlike the number variance, the variance ${\rm var}_\beta [\lambda_L]$
characterizes an ordered statistic.
 
Although ordinary and ordered spectral statistics carry fundamentally
different information -- the latter being considerably richer at the level
of their finite moments -- exact relations between the two are not
inconceivable. An instructive example is provided by the
uni-directional identity~(see Eq.~(6.1.39) in Ref.~\cite{M-2004})~\cite{Note-Mehta-Misprint}
\begin{eqnarray}\label{ex-fin-inf}
R_2(s) = \sum_{\ell=1}^{\infty} p_\beta(\ell;s)
\end{eqnarray}
connecting the two-point correlation function $R_2(s)$ -- an
ordinary statistic of {\it finite complexity} -- to the
probability densities $p_\beta(\ell;s)$ for the
distance $s$ from a typical reference eigenlevel to its $\ell$-th
neighbor to the right (ordered statistics).
Here and below, the correlation function and spacing
densities refer to the unfolded bulk limit $N\rightarrow\infty$,
with unit mean level density. Translational invariance in this
limit makes the two-point correlation function depend only on
the level separation $s$.
The mechanism underlying uni-directional relations of this type
is clear: Summation over the ordering mark $\ell$ effectively
restores the permutational invariance discussed
below Eq.~(\ref{ordinary-rv}).
 
Bi-directional identities can also emerge, provided that the ordinary
statistics exhibit infinite complexity. The relationship~(see Eq.~(6.1.18) in Ref.~\cite{M-2004})
\begin{eqnarray}\label{ex-inf-inf}
    p_\beta(k;s) = \frac{d^2}{ds^2}
    \sum_{\ell=0}^{k-1} (k-\ell) E_\beta(\ell;s),
    \; k\ge 1,
\end{eqnarray}
between the $k$-th neighbor spacing density $p_\beta(k;s)$ defined below
Eq.~(\ref{ex-fin-inf}) (ordered) and the counting probabilities $E_\beta(\ell;s)$ that a fixed
interval of length $s$ contains exactly $\ell$ eigenlevels (ordinary) serves as a classical example.

While the origin of the bi-directional formula Eq.~(\ref{ex-inf-inf}) and
analogous identities is well understood, there exists a candidate for a
somewhat intriguing bi-directional relation between the number variance
${\rm var}_\beta [{\mathcal N}(L)]$ -- {\it an ordinary statistic of finite
complexity} -- and the variance of the $L$-th {\it ordered} eigenvalue ${\rm
var}_\beta [\lambda_L]$ -- first proposed by French, Mello, and
Pandey~\cite{FMP-1978} nearly half a century ago. Using heuristic arguments,
they deduced the formula
\begin{eqnarray} \label{FMP-heuristics}
        {\rm var}_\beta [{\mathcal N}(L)] \approx  {\rm var}_\beta [\lambda_L] + \frac{1}{6}
\end{eqnarray}
which should hold ``to good precision''~\cite{BFFMPW-1981} for integers $L\gg 1$.

Even though the intuition behind a relation between these two
variances is clear -- spectral rigidity of the ${\rm Sine}_\beta$
process couples fluctuations of the counting function on the
scale $L$ to the displacement of the $L$-th ordered level --
turning this idea into a precise quantitative statement is not
straightforward. Early numerical verifications of Eq.~(\ref{FMP-heuristics}) indicated that
-- contrary to expectations -- the proposed relation appeared to hold better
for small values $1 \le L \le 5$ (see Section~V.C in
Ref.~\cite{BFFMPW-1981}) than for larger $L$, where the discrepancy between
the two variances was reported to increase, see Chapter~14 in
Ref.~\cite{PF-book}. These observations only added to the uncertainty
surrounding the validity of the identity. Indeed, its status remained
elusive even to the authors themselves, who acknowledged~\cite{BFFMPW-1981}:
\begin{quote}
{\it ``There are certain mysteries connected with [Eq.~(\ref{FMP-heuristics})], which [...] is of more consequence than might appear; it would be good to have a better derivation and understanding of it.''}
\end{quote}

In the present paper, we revisit the relation between these two spectral
variances. Transcending the heuristic arguments of the 1970s and employing a
technique based on our earlier work on the power spectrum description of
eigenvalue fluctuations in random matrix theory, not only do we prove that
Eq.~(\ref{FMP-heuristics}), taken at $\beta=2$, is exact in the limit
\begin{eqnarray} \label{TRK-asymptotic}
        \lim_{L\rightarrow \infty} \Big( {\rm var}_\beta [{\mathcal N}(L)] -  {\rm var}_\beta [\lambda_L] \Big) = \frac{1}{6},\quad \beta=2
\end{eqnarray}
but we also determine the explicit convergence trajectory to the universal
limit as $L \rightarrow\infty$. For the orthogonal ($\beta=1$) and
symplectic ($\beta=4$) symmetry classes, conjectural extensions of these
results are corroborated by a comprehensive, high-accuracy numerical
analysis.

The remainder of this paper is organized as follows.
Section~\ref{MRD-section} summarizes our main analytical results and
conjectures, including the asymptotic expansion of the spectral variance
difference (Theorem~\ref{main-T} and Conjecture~A), the asymptotic behavior
of the variance of the $L$-th ordered eigenvalue (Proposition~\ref{main-L}
and Conjecture~B), and the new sum rule for level-spacing auto-covariances
(Proposition~\ref{main-P} and Conjecture~C). Section~\ref{D-section}
presents a comprehensive numerical analysis. There, we use Bornemann’s
toolbox \texttt{RMTFredholm} for MATLAB~\cite{B-2010} -- based on the direct
quadrature discretization of Fredholm determinants -- to generate the
high-precision spectral data needed to test our analytical predictions and
conjectured convergence laws. The proofs are collected in
Section~\ref{proofs-section}. In Section~\ref{AC-section}, we establish the
sum rule for level-spacing auto-covariances. In
Section~\ref{Duality-section}, we derive the asymptotic expansion of the
variance of the $L$-th ordered eigenvalue and complete the proof of the
``$1/6$'' formula. Finally, Section~\ref{technical-section} provides the
auxiliary small-$\omega$ expansion of the counting-function generating
integral, supplying a lemma used in the proof of Proposition~\ref{main-P}.

\section{Main results}\label{MRD-section}

In this section, we present the analytical findings and conjectures that quantify the asymptotic relationship between global fluctuations of the counting function on the one hand and the local fluctuations of an individual level on the other. Our primary contribution, Theorem~\ref{main-T}, establishes the precise convergence law of the variance difference ${\rm var}_2 [{\mathcal N}(L)] -  {\rm var}_2 [\lambda_L]$ to the universal $1/6$ constant for the unitary symmetry class, utilizing an asymptotic expansion for the $L$-th ordered eigenvalue variance (Proposition~\ref{main-L}) and a previously unknown sum rule for level-spacing auto-covariances (Proposition~\ref{main-P}). Theorem~\ref{main-T} follows from Proposition~\ref{main-L}, whose proof in turn depends on Proposition~\ref{main-P}. All three results are formulated below, together with their conjectural extensions to the $\beta=1$ and $\beta=4$ symmetry classes (Conjectures~A, B and C).

\begin{theorem}\label{main-T}
  For large positive integers $L$, the following asymptotic expansion holds:
  \begin{eqnarray} \label{main-T-eq}\fl \qquad
  {\rm var}_2 [{\mathcal N}(L)] -  {\rm var}_2 [\lambda_L]  = \frac{1}{6}  + \frac{1}{2 \pi^4 L^2} \Big(
    \log(2 \pi L) + \gamma - \frac{\pi^2 + 9}{6}
  \Big) + o(L^{-2}),
\end{eqnarray}
where $\gamma=0.57721\dots$ is the Euler–Mascheroni constant.
\end{theorem}
\noindent
For a proof, see Section~\ref{Duality-section}.

\begin{conjectureA}\label{Conj-A}
For $\beta=1$ and $4$, we put forward the conjecture:
  \begin{eqnarray} \label{Conj-1-eq}
  {\rm var}_\beta [{\mathcal N}(L)] -  {\rm var}_\beta [\lambda_L]  = \frac{1}{6}- \frac{\delta_{\beta, 4}}{8\pi^2 L} +
  O\left(\frac{\log L}{L^2}\right).
\end{eqnarray}
\end{conjectureA}
\noindent
\begin{proposition}\label{main-L}
Let $\lambda_L$ be the $L$-th ordered eigenvalue in the infinite ordered spectrum, see Eq.~(\ref{L-th-ordered}). Then, as $L \rightarrow\infty$, its variance admits the asymptotic expansion:
\begin{eqnarray} \label{var-2-lambda-L} \fl\qquad\qquad
    {\rm var}_2[\lambda_L] = \frac{1}{\pi^2}\Big(\log (2\pi L) +\gamma+1
    \Big) -\frac{1}{6} \nonumber\\
    \fl\qquad\qquad\qquad\qquad
    - \frac{1}{2\pi^4 L^2}
    \Big(\log(2\pi L) + \gamma- \frac{\pi^2+6}{6}
    \Big) + o(L^{-2}).
\end{eqnarray}
\end{proposition}
\noindent
For a proof, see Section~\ref{Duality-section}.

\begin{conjectureB} \label{Conj-B}
Let $\lambda_L$ be the $L$-th ordered eigenvalue as defined in Proposition~\ref{main-L} and let
\begin{eqnarray} \label{v-beta}
    v_\beta = \left\{
    \begin{array}{cc}
      - \displaystyle\frac{\pi^2}{8}, & \beta=1; \\
      \log 2 + \displaystyle\frac{\pi^2}{8}, & \beta=4.
    \end{array}
    \right.
\end{eqnarray}
Then, as $L \rightarrow\infty$, we put forward the conjecture:
\begin{eqnarray} \label{var-beta-lambda-L}\fl\qquad\qquad
    {\rm var}_\beta[\lambda_L] = \frac{2}{\beta \pi^2}\Big(\log (2\pi L) + v_\beta +\gamma+1
    \Big) -\frac{1}{6} + O\left(
        \frac{\log L}{L^2}
    \right).
\end{eqnarray}
\end{conjectureB}

\begin{proposition}[Sum rule, $\beta=2$] \label{main-P}
Let $\delta I_{\ell}^{(2)}$ be the auto-covariances of level spacings in the ${Sine}_2$ process as defined in Eq.~(\ref{dI-ell}). Then, the following sum rule holds:
\begin{eqnarray}\label{L-1-2}
    C^{(1)}_2=\sum_{\ell=1}^{\infty} \ell \left( \delta I_\ell^{(2)} + \frac{1}{2 \pi^2\ell^2}  \right)
      = \frac{1}{12} -\frac{1}{2 \pi^2}\log(2\pi).
\end{eqnarray}
\end{proposition}\noindent
For a proof, see Section~\ref{AC-section}.

\begin{conjectureC}[Sum rule, $\beta=1$ and $4$] \label{Conj-C}
Let $\delta I_{\ell}^{(\beta)}$ be the auto-covariances of level spacings as defined in Eq.~(\ref{dI-ell}). Then, we conjecture the following sum rule:
\begin{eqnarray}\label{L-1-beta}
    C_\beta^{(1)}=\sum_{\ell=1}^{\infty} \ell \left( \delta I_\ell^{(\beta)} + \frac{1}{\beta \pi^2\ell^2}  \right)
      = \frac{1}{12} -\frac{1}{\beta \pi^2} \big( v_\beta + \log(2\pi)\big),
\end{eqnarray}
where $v_\beta$ is defined as in Eq.~(\ref{v-beta}).
\end{conjectureC}

The dichotomy in our presentation -- providing rigorous proofs for the unitary class ($\beta=2$) while leaving the orthogonal ($\beta=1$) and symplectic ($\beta=4$) classes as conjectures -- stems from the technical requirements of our analytical framework. Specifically, the proofs rely on precise control over the asymptotic behavior of the counting function cumulants. While these asymptotics are well understood for ${\beta=2}$~\cite{RK-2023}, see also Ref.~\cite{IAC-2013}, they remain significantly less studied for $\beta=1$ and $4$. Thus, we leave formal proofs of Conjectures A, B, and C for a separate study.

\section{Numerical analysis and discussion}\label{D-section}

\begin{figure}[t]
    \centering
    \includegraphics[width=0.85\textwidth]{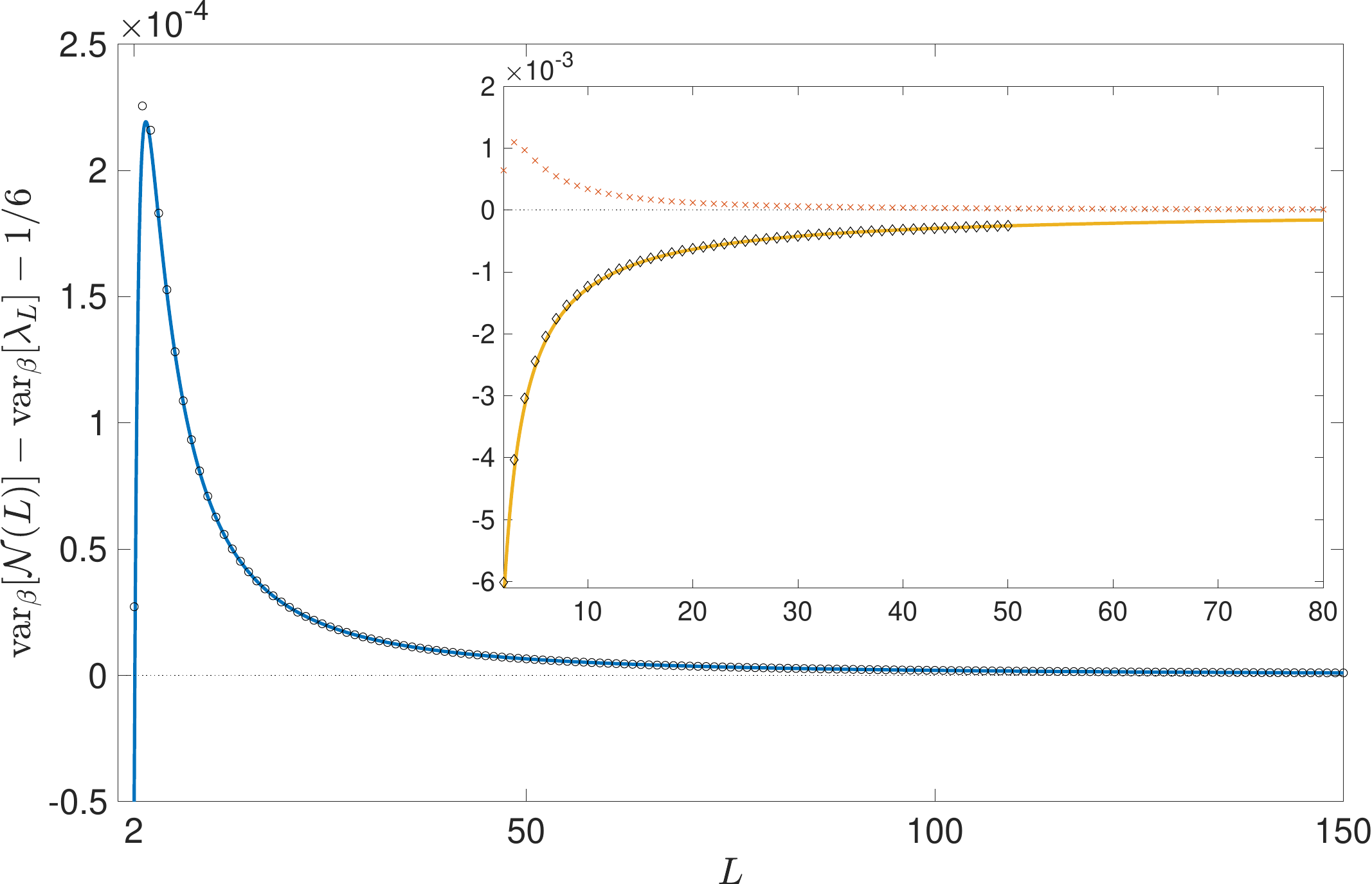}
    \caption{
        The deviation of the variance difference ${\rm var}_\beta[\mathcal{N}(L)]
        - {\rm var}_\beta[\lambda_L]$ from the universal $1/6$ limit as a function of $L$.
        The main panel showcases the $\beta=2$ convergence: the blue curve represents the theoretical correction in
        Eq.~(\ref{main-T-eq}) against numerical data (circles) generated via Bornemann's toolbox~\cite{B-2010}.
        Inset: numerical results for $\beta=1$ (crosses) and $\beta=4$ (diamonds), with the yellow curve marking the
        $-1/(8\pi^2L)$ correction [Eq.~(\ref{Conj-1-eq}), Conjecture~A] for the latter class. The negative sign of this 
        leading term explains the asymptotic approach to zero from below. No theoretical curve is 
        shown for $\beta=1$, since Conjecture~A provides only an $O(L^{-2}\log L)$ estimate for the correction.
        }
    \label{one-six}
\end{figure}

Before providing the proofs in the subsequent sections, we turn to a detailed numerical validation of our theoretical predictions.
This is necessitated by the ambiguities in earlier numerical investigations of the $1/6$ formula, as discussed in the Introduction.
By employing high-precision numerical computations using Bornemann's MATLAB toolbox \texttt{RMTFredholm}~\cite{B-2010}, whose routines
handle kernels already taken in the unfolded bulk limit $N\to\infty$, we verify the universal asymptotic limit, the predicted symmetry-dependent 
convergence laws, and the previously unknown sum rule across the three Dyson symmetry classes. For $\beta=4$, the toolbox generates the
ordered-level variances from the $\beta=1$ data using relations between the orthogonal and symplectic ensembles (see Appendix~A in Ref.~\cite{B-2010}).
Since these relations require $\beta=1$ data at twice the value of $L$, the $\beta=4$ computations cover a smaller range of $L$ than those for the 
$\beta=1$ and $2$ symmetry classes, as reflected in the figures and tables below.
\noindent\newline\newline
{\it Theorem~\ref{main-T} and Conjecture~A: Figure~\ref{one-six}.}---To validate our analytical framework, Figure~\ref{one-six} plots the numerically evaluated deviation of $\mathrm{var}_\beta[\mathcal{N}(L)] - \mathrm{var}_\beta[\lambda_L]$ from the $1/6$ plateau. The results confirm that while the asymptotic duality is universal, the approach to the limit is diverse and varies significantly between classes.

\begin{table}[b]
\centering
\small
\renewcommand{\arraystretch}{1.2}
\caption{Numerical evaluation of the variance difference $\Delta_\beta(L) = \mathrm{var}_\beta[\mathcal{N}(L)] - \mathrm{var}_\beta[\lambda_L]$ for selected values of $L$ for the $\beta=1, 2$ and $4$ symmetry classes. For $\beta=4$, we also focus on the shifted variance difference $\Delta_4^* = \Delta_4 + (8\pi^2 L)^{-1}$, demonstrating how the leading $O(L^{-1})$ correction term proposed in Conjecture~A accounts for the primary deviation from the $1/6$ asymptotic limit.}
\label{one-six-table}

\begin{tabular}{c S[table-format=1.8] S[table-format=1.2e-1] S[table-format=1.8] S[table-format=1.2e-1]}
\toprule
& \multicolumn{2}{c}{$\beta = 1$} & \multicolumn{2}{c}{$\beta = 2$} \\
\cmidrule(lr){2-3} \cmidrule(lr){4-5}
$L$ & {$\Delta_1$} & {$\Delta_1 -1/6$} & {$\Delta_2$} & {$\Delta_2 - 1/6$} \\
\midrule
2   & 0.16731036 & 6.44e-4 & 0.16669386 & 2.72e-5 \\
5   & 0.16746704 & 8.00e-4 & 0.16684974 & 1.83e-4 \\
10  & 0.16700836 & 3.42e-4 & 0.16674765 & 8.10e-5 \\
20  & 0.16678762 & 1.21e-4 & 0.16669577 & 2.91e-5 \\
30  & 0.16672966 & 6.30e-5 & 0.16668191 & 1.52e-5 \\
40  & 0.16670579 & 3.91e-5 & 0.16667616 & 9.50e-6 \\
50  & 0.16669354 & 2.69e-5 & 0.16667320 & 6.53e-6 \\
60  & 0.16668636 & 1.97e-5 & 0.16667146 & 4.80e-6 \\
70  & 0.16668178 & 1.51e-5 & 0.16667035 & 3.69e-6 \\
80  & 0.16667867 & 1.20e-5 & 0.16666960 & 2.93e-6 \\
90  & 0.16667645 & 9.78e-6 & 0.16666906 & 2.39e-6 \\
100 & 0.16667481 & 8.14e-6 & 0.16666866 & 1.99e-6 \\
\bottomrule
\end{tabular}

\vspace{2em}

\begin{tabular}{c S[table-format=1.8] S[table-format=-1.2e-1] S[table-format=1.8] S[table-format=1.2e-1]}
\toprule
& \multicolumn{4}{c}{$\beta = 4$ } \\
\cmidrule(lr){2-5}
$L$ & {$\Delta_4$} & {$\Delta_4 - 1/6$} & {$\Delta_4^{*}$} & {$\Delta_4^{*} - 1/6$} \\
\midrule
2   & 0.16065153 & -6.02e-3 & 0.16698410 & 3.17e-4 \\
5   & 0.16422413 & -2.44e-3 & 0.16675716 & 9.05e-5 \\
10  & 0.16543103 & -1.24e-3 & 0.16669754 & 3.09e-5 \\
20  & 0.16604327 & -6.23e-4 & 0.16667653 & 9.86e-6 \\
30  & 0.16624944 & -4.17e-4 & 0.16667161 & 4.95e-6 \\
40  & 0.16635305 & -3.14e-4 & 0.16666968 & 3.01e-6 \\
50  & 0.16641540 & -2.51e-4 & 0.16666870 & 2.04e-6 \\
\bottomrule
\end{tabular}
\end{table}

\begin{figure}[t]
    \centering
    \includegraphics[width=0.85\textwidth]{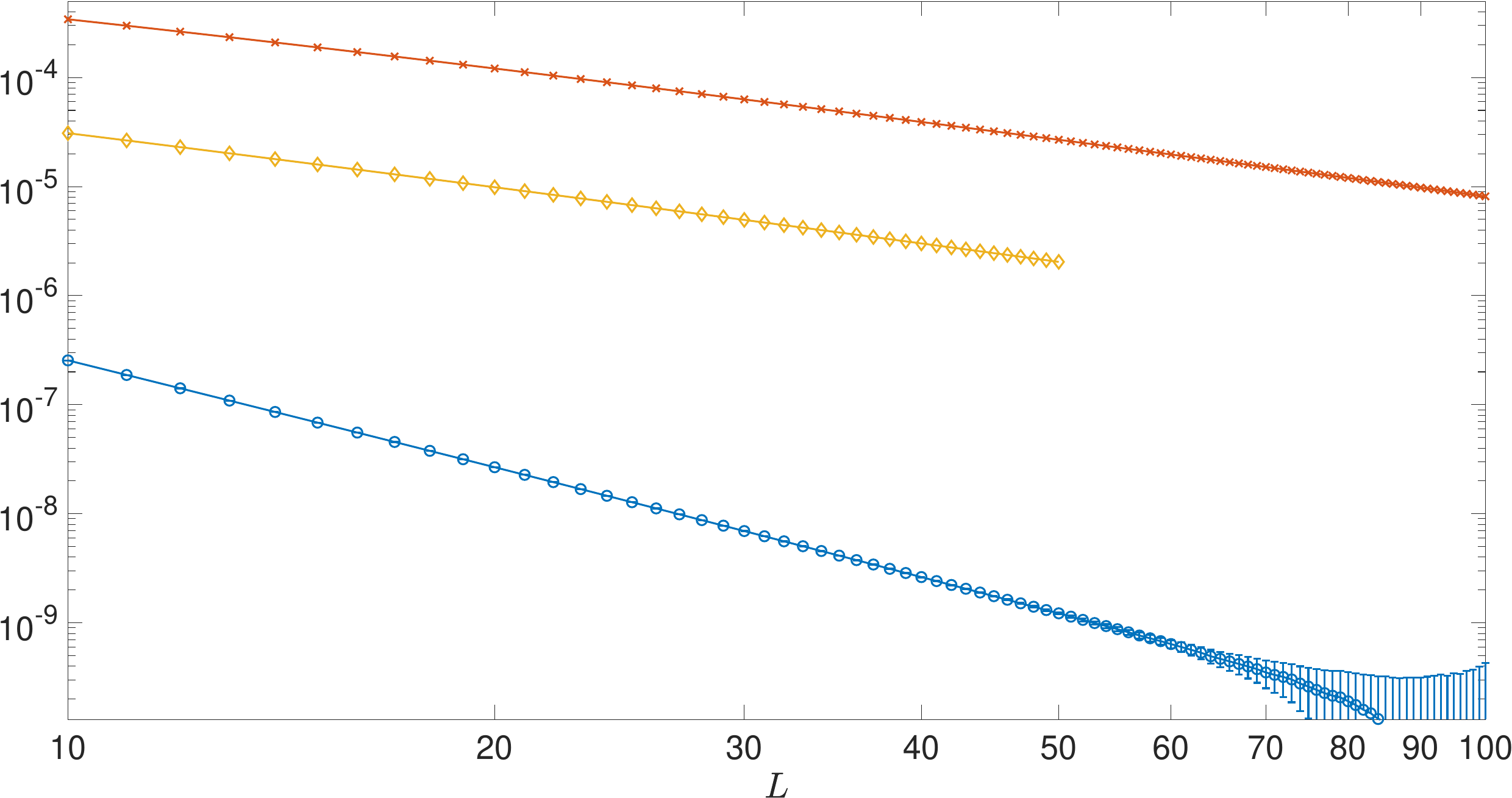}
    \caption{
    Residual difference obtained after subtracting from the numerically computed difference $\mathrm{var}_\beta[\mathcal{N}(L)]-\mathrm{var}_\beta[\lambda_L]$
    the asymptotic expressions stated in Eq.~(\ref{main-T-eq}) for $\beta=2$ and in Eq.~(\ref{Conj-1-eq}) for $\beta=1, 4$. The red crosses, blue circles, and yellow
    diamonds correspond to $\beta=1, 2, 4$, respectively. In the log-log representation, the residuals for $\beta=1$ and $\beta=4$ exhibit an approximate slope $-2$,
    whereas the residual for $\beta=2$ exhibits an approximate slope $-4$ until numerical saturation sets in for $L\gtrsim 60$. Error bars indicate the estimated numerical
    uncertainty from Bornemann's toolbox.
    }
    \label{one-six-errors}
\end{figure}

\begin{itemize}
\item The $\beta=2$ profile displayed in the main panel reveals close agreement between the numerical circles
    and the theoretical blue curve. The expansion derived in Eq.~(\ref{main-T-eq}) remains robust even for small $L$,
    illustrating that the approach to the limit is governed by subtle corrections with the decay rates
    $O(L^{-2}\log L)$ and $O(L^{-2})$. These results clarify the large-$L$ convergence
    that remained uncertain in earlier numerical studies~\cite{BFFMPW-1981,PF-book}, as discussed in the Introduction.

\item As shown in the inset, the $\beta=4$ and $\beta=1$ cases exhibit different behaviors. For $\beta=4$ (diamonds), the approach to the plateau is notably slower and is well captured by the $O(L^{-1})$ leading correction given in Conjecture~A. This term represents the dominant deviation for the symplectic class. In contrast, the $\beta=1$ case (crosses) shows a faster convergence, consistent with the $O(L^{-2}\log L)$ bound in Eq.~(\ref{Conj-1-eq}). The absence of a $1/L$ term for $\beta=1$ suggests a structural similarity to the $\beta=2$ case.
\end{itemize}
\noindent\newline
{\it Theorem~\ref{main-T} and Conjecture~A: Table~\ref{one-six-table}.}---These observations are quantified in Table~\ref{one-six-table}, which provides a detailed evaluation of the spectral variance differences $\Delta_\beta(L) = \mathrm{var}_\beta[\mathcal{N}(L)] - \mathrm{var}_\beta[\lambda_L]$ across the three Dyson symmetry classes.
\begin{itemize}
\item For the $\beta=2$ case, numerical data substantiate the analytical expansion in Theorem~\ref{main-T} with high precision. Although discrepancies $\Delta_2-1/6$ for $L < 10$ lie in the $10^{-4}$--$10^{-5}$ range, the residual narrows to $6.53 \times 10^{-6}$ at $L=50$ and reaches $1.99 \times 10^{-6}$ by $L=100$. This agreement confirms that the expansion in Eq.~(\ref{main-T-eq}) describes eigenvalue fluctuations fairly well even before the asymptotic regime is reached.

\item The $\beta=1$ convergence profile follows a trajectory similar to the unitary case, supporting the absence of a slow $1/L$ leading-order term in Eq.~(\ref{Conj-1-eq}). Despite initial deviations $\Delta_1-1/6$ for small $L$ that stay within the $10^{-3}$--$10^{-4}$ range, the accuracy improves steadily as the system approaches the plateau. Specifically, the error reduces to $2.69 \times 10^{-5}$ at $L=50$ and falls further to $8.14 \times 10^{-6}$ by $L=100$. Such close alignment provides strong empirical evidence for the $O(L^{-2} \log L)$ proposed in Conjecture~A.

\item In the $\beta=4$ case, the raw variance difference $\Delta_4$ exhibits the slowest relaxation, yet the $1/L$ correction proposed in Conjecture~A effectively accounts for this discrepancy. After including this term, the residual deviation $\Delta_4^* -1/6$ for $L < 10$ is already reduced to the $10^{-4}$--$10^{-5}$ range, and by $L=50$, it is reduced to $2.04 \times 10^{-6}$. This order-of-magnitude improvement provides strong numerical support for the symmetry-dependent corrections and the special standing of the symplectic class as formulated in Conjecture~A.
\end{itemize}
\begin{figure}[t]
    \centering
    \includegraphics[width=0.85\textwidth]{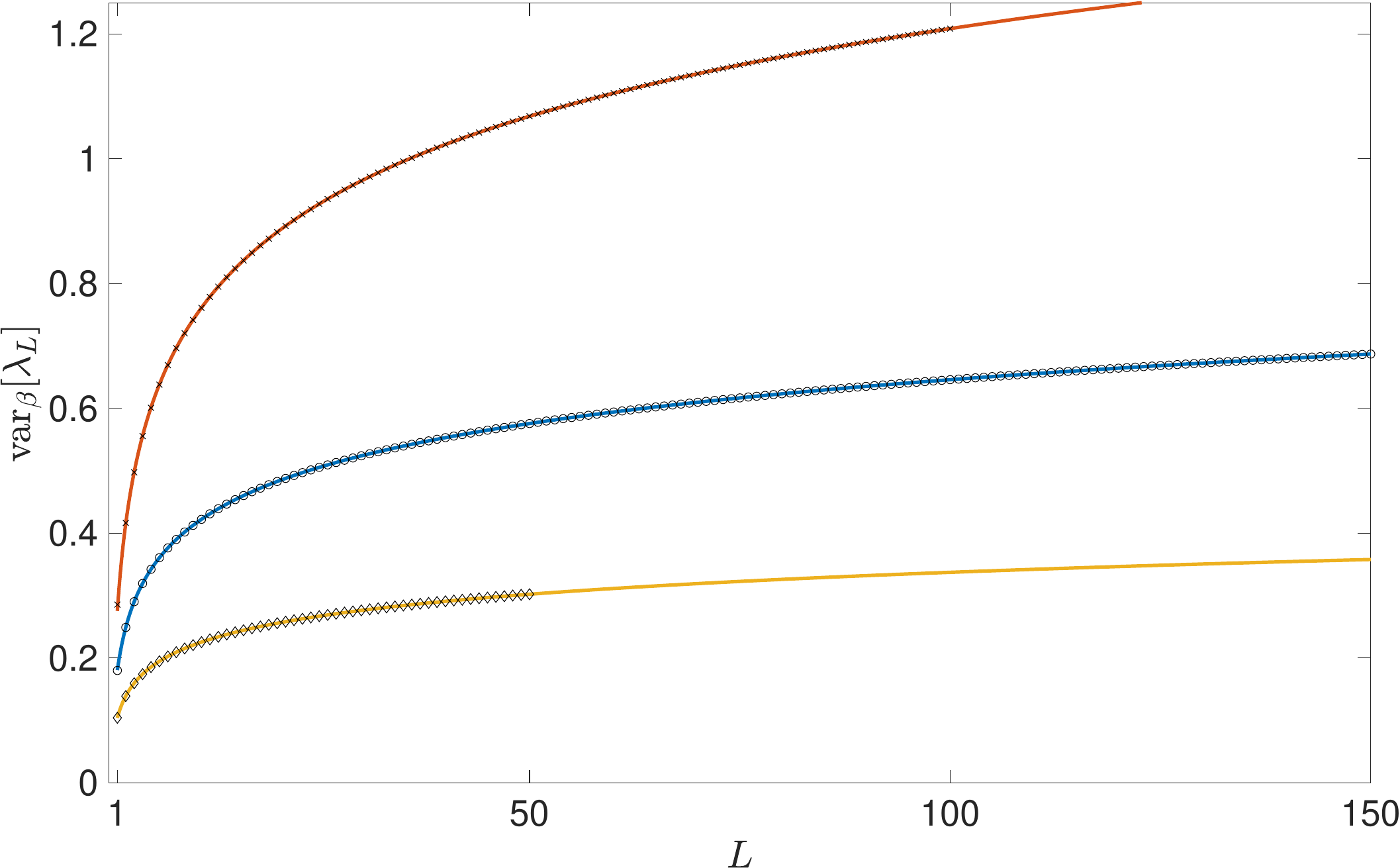}
    \caption{
    The variance $\text{var}_{\beta}[\lambda_L]$ of the $L$-th ordered eigenvalue as a function of $L$. The theoretical curve for $\beta=2$ (blue line) corresponds to Eq.~(\ref{var-2-lambda-L}) of Proposition~\ref{main-L}, while those for $\beta=1$ (red line) and $\beta=4$ (yellow line) symmetry classes reflect Eq.~(\ref{var-beta-lambda-L}) of Conjecture~B. Numerical data obtained with the help of Bornemann's toolbox~\cite{B-2010} are denoted by circles ($\beta=2$), crosses ($\beta=1$), and diamonds ($\beta=4$).
    }
    \label{ordered-eig}
\end{figure}
\noindent
{\it Theorem~\ref{main-T} and Conjecture~A: Figure~\ref{one-six-errors}.}---Whereas Figure~\ref{one-six} and Table~\ref{one-six-table} establish the convergence to the
universal $1/6$ limit, Figure~\ref{one-six-errors} directly tests the quality of the asymptotic laws themselves. After subtraction of the displayed terms in
Eq.~(\ref{main-T-eq}) and Eq.~(\ref{Conj-1-eq}), the residuals for $\beta=1$ and $\beta=4$ follow approximately straight lines of slope $-2$, consistent with the
stated $O(L^{-2}\log L)$ control. In the unitary class, the residual has slope about $-4$, in agreement with the sharper expansion in Theorem~\ref{main-T}, until the
curve eventually reaches the numerical precision floor of Bornemann's toolbox near $L\approx 60$.
\noindent\newline\newline
{\it Proposition~\ref{main-L} and Conjecture~B: Figure~\ref{ordered-eig}.}---Figure~\ref{ordered-eig} illustrates the variance $\mathrm{var}_{\beta}[\lambda_L]$ across three symmetry classes. The red ($\beta=1$), blue ($\beta=2$), and yellow ($\beta=4$) curves are differentiated by symmetry-dependent offsets $v_\beta$ defined in Eq.~(\ref{v-beta}). The blue curve plots the expansion in Eq.~(\ref{var-2-lambda-L}) from Proposition~\ref{main-L}, while the red and yellow lines represent the predictions of Conjecture~B. All three curves show a remarkable fit to high-precision numerical data obtained with Bornemann's toolbox -- indicated by symbols -- offering high-precision numerical corroboration of the analytical results for $\beta=2$ and strong numerical support for the conjectural extensions. Strikingly, although these expansions were derived in the asymptotic $L \rightarrow \infty$ limit, they remain robust and fit the numerical data even at $L=1$, describing eigenvalue fluctuations well before the true asymptotic regime is reached.
\noindent\newline\newline
{\it The sum rule: Proposition~\ref{main-P} and Conjecture~C.}---The theoretical values of the sum rule constants $C_\beta^{(1)}$ -- proven in Proposition~\ref{main-P} for the unitary class ($\beta=2$) and predicted by Conjecture~C for the orthogonal ($\beta=1$) and symplectic ($\beta=4$) classes -- are defined as:
\begin{equation} \label{c-beta-summary}
    C_\beta^{(1)} =
    \begin{cases}
      \displaystyle\frac{5}{24}- \frac{1}{\pi^2} \log(2\pi), & \beta=1; \\[2ex]
      \displaystyle\frac{1}{12} - \frac{1}{2\pi^2} \log(2\pi), & \beta=2; \\[2ex]
      \displaystyle\frac{5}{96} - \frac{1}{4\pi^2} \log(4\pi), & \beta=4.
    \end{cases}
\end{equation}
To compare these constants against numerical values, we first transform their {\it series} definition Eqs.~(\ref{L-1-2}) and (\ref{L-1-beta}) into equivalent form that is more suitable for numerical analysis. To this end, we split the series representation of $C_\beta^{(1)}$ into two parts,
\begin{eqnarray}\label{C-beta-split}
    C_\beta^{(1)} = \sum_{\ell=1}^{M-1} \ell \left(\delta I_\ell^{(\beta)}+\frac{1}{\beta\pi^2\ell^2}\right) + \sum_{\ell=M}^{\infty} \ell \left(\delta I_\ell^{(\beta)}+\frac{1}{\beta\pi^2\ell^2}\right),
\end{eqnarray}
and make use of the identity~\cite{RTK-2023,Note-Reference-Level}
\begin{eqnarray} \label{sum-M-1}
    \delta I_\ell^{(\beta)} = \frac{1}{2} \left(
        {\rm var}_\beta [\lambda_{\ell+1}] - 2 {\rm var}_\beta [\lambda_{\ell}] + {\rm var}_\beta [\lambda_{\ell-1}]
    \right).
\end{eqnarray}
By virtue of Eq.~(\ref{sum-M-1}), the first sum in Eq.~(\ref{C-beta-split}) reduces to
\begin{eqnarray} \label{sum-M-2} \fl\qquad
    \sum_{\ell=1}^{M-1} \ell \left(\delta I_\ell^{(\beta)}+\frac{1}{\beta\pi^2\ell^2}\right)
    = \frac{M-1}{2} {\rm var}_\beta [\lambda_{M}] - \frac{M}{2} {\rm var}_\beta [\lambda_{M-1}] + \frac{1}{\beta\pi^2} H_{M-1},
\end{eqnarray}
where $H_k$ is the $k$-th harmonic number thus yielding a numerically friendly variant of Eq.~(\ref{C-beta-split}):
\begin{eqnarray}\label{C-beta-split-1}\fl\qquad\qquad
    C_\beta^{(1)} = \frac{M-1}{2} {\rm var}_\beta [\lambda_{M}] - \frac{M}{2} {\rm var}_\beta [\lambda_{M-1}] + \frac{1}{\beta\pi^2} H_{M-1} \nonumber\\
\fl\qquad\qquad\qquad\qquad\qquad\qquad\qquad + \sum_{\ell=M}^{\infty} \ell \left(\delta I_\ell^{(\beta)}+\frac{1}{\beta\pi^2\ell^2}\right).
\end{eqnarray}
The variances in Eq.~(\ref{C-beta-split-1}) can be produced, with high precision, by Bornemann's toolbox~\cite{B-2010}. In addition, assuming that $M$ is large enough, the remaining series can be carefully estimated from Eq.~(\ref{beyond-Dyson}) for $\beta=2$.

\begin{table}[t]
\centering
\small
\renewcommand{\arraystretch}{1.3}
\caption{Numerical verification of the sum rule constants $C_{\beta}^{(1)}$ across the Dyson symmetry classes. Theoretical values are given by Eq.~(\ref{c-beta-summary}).}
\label{sum-rule-table}
\begin{tabular}{l S[table-format=-1.9] S[table-format=-1.9] S[table-format=-1.9]}
\toprule
& {$\beta=1$} & {$\beta=2$} & {$\beta=4$} \\
\midrule
$C_{\beta}^{(1)}$ (theory)   & 0.022117454 & -0.009774607 & -0.012028260 \\
$C_{\beta}^{(1)}$ (numerics) & 0.022130516    & -0.009774605     & -0.012023026    \\
Error                        & \num{1.31e-05} & \num{1.33e-09}   & \num{5.23e-06}  \\
$M$                          & {100}          & {100}            & {50}            \\
\bottomrule
\end{tabular}
\end{table}

\begin{figure}[t]
    \centering
    \includegraphics[width=0.85\textwidth]{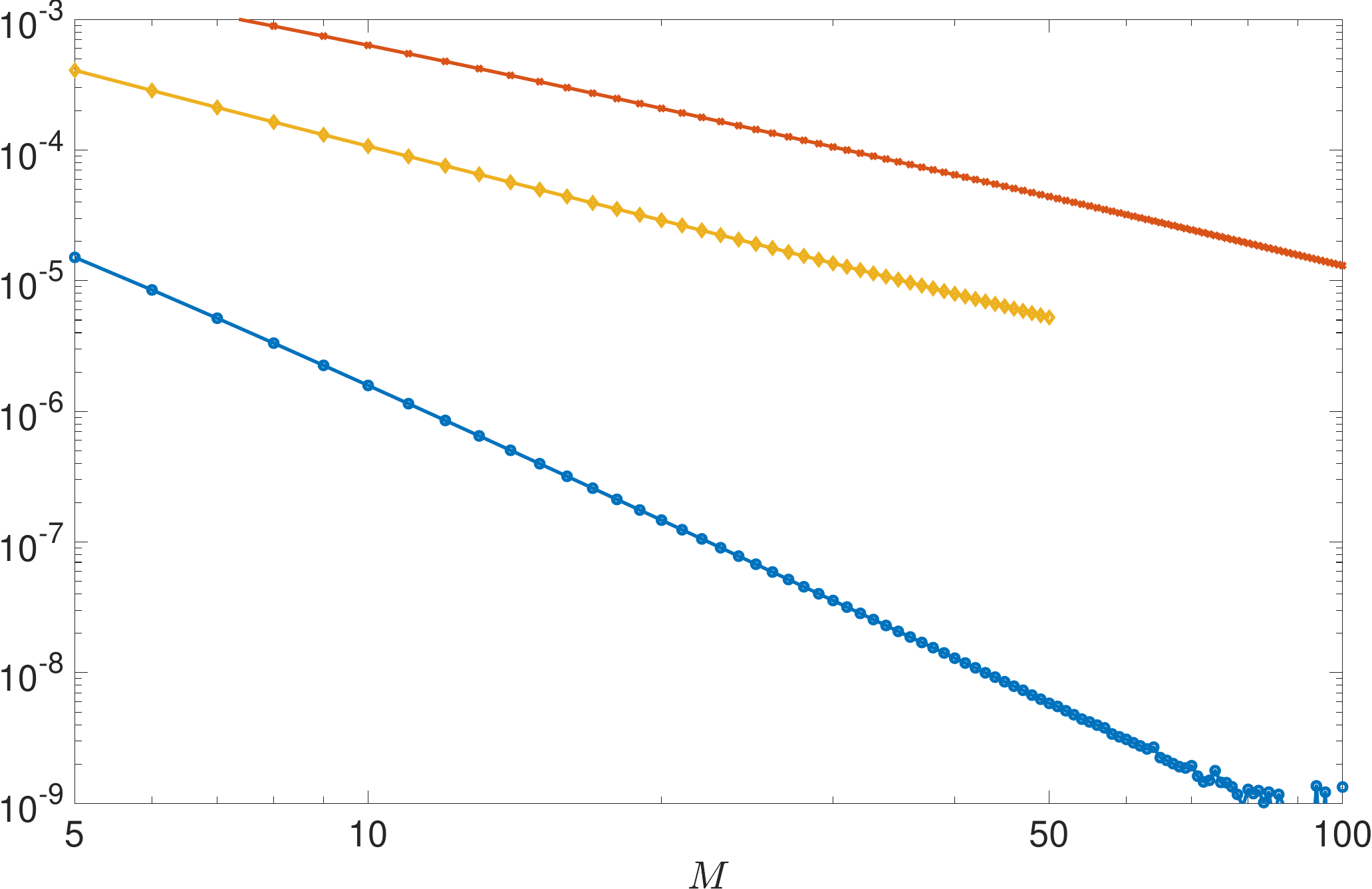}
    \caption{Difference of numerically evaluated values of $C_\beta^{(1)}$, Eq.~(\ref{C-beta-split-1}), and their theoretical values, Eq.~(\ref{c-beta-summary}), for
     $\beta=1$ (red crosses), $\beta=2$ (blue circles), and $\beta=4$ (yellow diamonds) as a function of the splitting parameter $M$.
    }
    \label{c-beta-error}
\end{figure}
The data presented in Table~\ref{sum-rule-table} indicate excellent agreement between the analytical predictions and the numerical results obtained via Bornemann’s toolbox. The agreement in the unitary class ($\beta=2$) is particularly striking, with a residual error on the order of $10^{-9}$; this clearly confirms the analytical derivation presented in Proposition~\ref{main-P}. A high level of consistency is also observed for the orthogonal ($\beta=1$) and symplectic ($\beta=4$) classes, where the numerical results yield errors of approximately $10^{-5}$ and $10^{-6}$, respectively. While these values are slightly less precise than the unitary case -- due to the disregarded tail contribution in Eq.~(\ref{C-beta-split-1}) -- the results remain well within expected bounds and provide strong numerical evidence for Conjecture~C.

The convergence behavior of the sum rule constants $C_{\beta}^{(1)}$ is illustrated in Figure~\ref{c-beta-error} as a function of the splitting parameter $M$. For the unitary class ($\beta=2$), the difference between the numerical and theoretical values follows an approximately straight line with a slope of about $-4$, reaching a residual error of $10^{-9}$ near $M=80$. This high level of precision is facilitated by the use of an asymptotic approximation for the infinite sum in Eq.~(\ref{C-beta-split-1}); as $M$ approaches 100, the error plateau suggests that the limit of numerical accuracy for evaluating the variances in Eq.~(\ref{C-beta-split-1}) has been reached. In contrast, the orthogonal ($\beta=1$) and symplectic ($\beta=4$) classes exhibit a shallower convergence slope of approximately $-2$. For these symmetries, the accuracy reaches a notable level of $10^{-5}$ at $M=100$ for $\beta=1$ and $\num{5e-06}$ at $M=50$ for $\beta=4$, despite omitting the tail contributions in Eq.~(\ref{C-beta-split-1}).
\noindent\newline\newline
{\it Summary.}---The numerical results presented here independently verify the main findings announced in Section 2 and resolve the ambiguities surrounding the validity of the $1/6$ formula. With these results corroborating both the universal asymptotic limit and the symmetry-dependent convergence rates across all three Dyson classes, we now proceed to the formal analytical derivations in the subsequent sections.

\section{Proof of main results}\label{proofs-section}

We construct the proofs from the ground up: Section~\ref{AC-section} establishes the sum rule (Proposition~\ref{main-P}), which is then utilized in Section~\ref{Duality-section} to derive the asymptotic variance of the ordered eigenvalue (Proposition~\ref{main-L}) and ultimately prove Theorem~\ref{main-T}. To streamline the exposition, the technical small-frequency expansions required for the sum rule are deferred to Section~\ref{technical-section}.

\subsection{Auto-covariances of level spacings and the sum rule}\label{AC-section}

{\it Definition and Dyson's conjecture.}---In what follows, we focus on unfolded bulk spectra of infinite-dimensional random matrices described by the ${\rm Sine}_\beta$ point process~\cite{KS-2009,VV-2017,MNN-2019}. To avoid introducing a biased spacing~\cite{TRK-2024}, we choose a level at random from this process, denote it as $\lambda_0$, and consider the infinite sequence of ordered energy levels $\{\dots, \lambda_{-1}, \lambda_0, \lambda_1,\dots\}$ such that $\lambda_{k}<\lambda_{k+1}$ for all $k$. Let $\{ \dots, s_{-1}, s_0, s_1,\dots\}$ be the sequence of associated level spacings $s_k = \lambda_{k} - \lambda_{k-1}$, and let
\begin{eqnarray}\label{dI-ell}
    \delta I_\ell^{(\beta)} = {\rm cov}(s_k,s_{k+\ell}) = {\mathbb E} \big[ s_k s_{k+\ell}\big] -1
\end{eqnarray}
denote their covariance matrix. The translational invariance of the ${\rm Sine}_\beta$ process and the random-level sampling ensure that the discrete sequence of spacings $\{s_k\}$ is strictly stationary~\cite{M-2004,VV-2017}. Consequently, the covariance matrix is independent of the index $k$ of the reference spacing.

The correlations between the spacings weaken as the distance between eigenlevels grows. The rate of correlation decay was conjectured by Dyson (unpublished) and quoted in Odlyzko's landmark numerical studies of the Riemann zeta function zeros~\cite{O-1987}:
\begin{eqnarray} \label{D-conjecture}
    \delta I_{|\ell|}^{(\beta)} \simeq \delta I_{|\ell|}^{(\beta, \rm D)} := - \frac{1}{\beta \pi^2 \ell^2},
\end{eqnarray}
where the superscript ${\rm D}$ refers to Dyson's conjecture. This power-law decay of distant spacings, ultimately related to the logarithmic rigidity of random-matrix-theory spectra, is assumed to hold asymptotically for sufficiently large $\ell$. For a heuristic argument in favor of Dyson's conjecture, the reader is referred to Ref.~\cite{BFFMPW-1981}.
\noindent\newline\newline
{\it Beyond Dyson.}---A nonperturbative approach to studying the auto-covariances of level spacings was formulated by the present authors, exploiting the fact that the power spectrum of level spacings is the discrete-time Fourier series of the auto-covariances~\cite{RTK-2023}:
\begin{eqnarray}\label{PS-A-FT}
    S^{\rm sp}_{\beta, \infty} (\omega)  = \sum_{\ell \in {\mathbb Z}} \delta I_{|\ell|}^{(\beta)} e^{i \omega \ell},\; \omega \in [0,\pi].
\end{eqnarray}
Since the power spectrum $S^{\rm sp}_{\beta,\infty} (\omega)$ is known exactly in terms of Painlev\'e transcendents (fifth Painlev\'e~\cite{RK-2023} for $\beta=2$ and third Painlev\'e~\cite{FW-2025} for $\beta=1$ and $4$), the inverse of Eq.~(\ref{PS-A-FT}) yields an exact formula for the auto-covariances
\begin{eqnarray}\label{ac-exact}
    \delta I_{|\ell|}^{(\beta)} = \frac{1}{\pi} \int_{0}^{\pi} d\omega\, S^{\rm sp}_{\beta, \infty} (\omega) \cos(\ell \omega).
\end{eqnarray}
This relation paves the way for a systematic study of the asymptotic behavior of the auto-covariances as $\ell \rightarrow\infty$ which is governed by the small-$\omega$ expansion of the power spectrum. For $\beta=2$, it holds~\cite{RTK-2023}
\begin{eqnarray} \label{beyond-Dyson}
 \delta I_{|\ell|}^{(2)} = -\frac{1}{2\pi^2\ell^2} - \frac{3}{2 \pi^4 \ell^4} \left(
    \log (2\pi \ell) + \gamma - \frac{11}{6}
 \right) + o(\ell^{-4}).
\end{eqnarray}
A high-precision numerical verification of this result can be found in Ref.~\cite{RTK-2023}. For heuristic treatment of $\beta=1$ and $4$ symmetry classes, the reader is referred to Ref.~\cite{FW-2025}.
\noindent\newline\newline
{\it Sum rules.}---Due to spectral rigidity, the auto-covariances $\delta I_{\ell}^{(\beta)}$ are tightly constrained: any local spacing fluctuation is compensated by anti-correlated fluctuations in the surrounding spacings. Quantitatively, this is expressed by Pandey's sum rule~\cite{P-1986}:
\begin{eqnarray} \label{Pandey}
   C^{(0)}_\beta = \sum_{\ell \in {\mathds Z}} \delta I_{|\ell|}^{(\beta)} = 0.
\end{eqnarray}
In terms of the power spectrum of level spacings [Eq.~(\ref{PS-A-FT})], this sum rule corresponds to the vanishing condition $S_{\beta, \infty}^{\rm sp}(0)=0$. This is indeed satisfied, as seen from the small-frequency expansion~\cite{ROK-2017,ROK-2020,RK-2023,FW-2025}
\begin{eqnarray} \label{small-omega}
    S_{\beta, \infty}^{\rm sp}(\omega)=\frac{|\omega|}{\beta \pi} + o(\omega).
\end{eqnarray}

Interestingly, Pandey's sum rule is not the only constraint on the auto-covariances of level spacings. In particular, the previously unknown sum rule formulated in Proposition~\ref{main-P} plays a central r\^ole in the proof of Theorem~\ref{main-T} and, consequently, in establishing the asymptotic duality of spectral variances, Eq.~(\ref{TRK-asymptotic}).

The proof of the sum rule in Eq.~(\ref{L-1-2}) relies on the following lemma. Although our primary focus is confined to $\beta=2$, we maintain a general presentation, in particular accommodating $\beta=1$ and $4$ unless stated otherwise.

\begin{lemma} \label{Fourier-lemma}
  Let $\delta I_\ell^{(\beta)}$ be the covariance matrix [Eq.~(\ref{dI-ell})] of level spacings in the ${\it Sine}_\beta$ point process with unit local density, $\rho=1$, and let $n(\lambda)$ be its point counting function (we have omitted the subscript $\beta$ for brevity). Then, for $0<\omega<2\pi$, the sine and cosine Fourier series of the covariance matrix are given by
  \begin{eqnarray} \label{cos-F}
    \frac{1}{2}\delta I_0^{(\beta)} + \sum_{\ell=1}^{\infty} \delta I_\ell^{(\beta)} \cos (\omega \ell) = {\rm Re\,} \int_{0}^{\infty} d\lambda \, \mathbb{E} \big[ e^{i\omega n(\lambda)} \big],
  \end{eqnarray}
  and
  \begin{eqnarray} \label{sin-F}
    \sum_{\ell=1}^{\infty} \delta I_\ell^{(\beta)} \sin (\omega \ell) = {\rm Im\,} \int_{0}^{\infty} d\lambda \, \mathbb{E} \big[ e^{i\omega n(\lambda)} \big] -\frac{1}{2} \cot\left( \frac{\omega}{2}\right),
  \end{eqnarray}
respectively.
\end{lemma}

\begin{proof}
  The moment generating function of the point counting function $n(\lambda)$ inside the open unit disk, $|z|<1$, can be written as
  \begin{eqnarray} \label{MGF}
    \mathbb{E} \big[ z^{n(\lambda)} \big] = \sum_{\ell=0}^{\infty} z^\ell E_\beta(\ell;\lambda),
  \end{eqnarray}
  where $E_\beta(\ell;\lambda)$ is the probability that an interval of length $\lambda$ contains exactly $\ell$ levels. Integrating over $\lambda \ge 0$ and using the relation~\cite{M-2004,TRK-2024}
  \begin{eqnarray} \label{cov-matrix-probs}
    (1+\delta_{\ell,0}) \int_{0}^{\infty} d\lambda\, E_\beta(\ell;\lambda) =  \delta I_\ell^{(\beta)} + 1,
  \end{eqnarray}
  where $\delta I_\ell^{(\beta)}$ denotes the covariance matrix of level spacings [Eq.~(\ref{dI-ell})], we obtain:
\begin{eqnarray} \label{complex-series}
    \frac{1}{2} \delta I_0^{(\beta)} +\sum_{\ell=1}^{\infty} z^\ell \delta I_\ell^{(\beta)} = \int_{0}^{\infty} d\lambda\, \mathbb{E} \big[ z^{n(\lambda)} \big] - \frac{1+z}{2(1-z)}.
\end{eqnarray}
Equation (\ref{complex-series}) is valid in the interior of the unit disk, $z =r e^{i\omega}$ with $0\le r < 1$. Taking the limit $r \rightarrow 1^-$ to continuously extend the result to the boundary of the unit disk, $|z|=1$, and separating the real and imaginary parts of Eq.~(\ref{complex-series}), completes the proof.

\end{proof}
\begin{remark}
Pandey's sum rule [Eq.~(\ref{Pandey})] follows from Eq.~(\ref{cos-F}) after recognizing that its r.h.s. equals one half of the power spectrum of spacings $S^{\rm sp}_{\beta,\infty} (\omega)$. Taking the limit $\omega\rightarrow 0$ in this relation and utilizing the small-$\omega$ expansion in Eq.~(\ref{small-omega}), we reproduce Eq.~(\ref{Pandey}).
\end{remark}

\begin{proof}[Proof of Proposition~\ref{main-P}]
The sine Fourier series established in Lemma~\ref{Fourier-lemma} serves as the starting point. Rewriting it in the form
\begin{eqnarray} \label{sin-F-reg}\fl\quad
    \sum_{\ell=1}^{\infty} \left( \delta I_\ell^{(\beta)} + \frac{1}{\beta \pi^2 \ell^2}\right) \sin (\omega \ell)
    = {\rm Im\,} \int_{0}^{\infty} d\lambda \, \mathbb{E} \big[ e^{i\omega n(\lambda)} \big]
    -\frac{1}{2} \cot\left( \frac{\omega}{2}\right)
    + \frac{1}{\beta \pi^2} {\rm Cl}_2(\omega), \nonumber\\
{}
\end{eqnarray}
where
\begin{eqnarray}
{\rm Cl}_2(\omega) = - \int_{0}^{\omega} dt \log \left| 2 \sin \left(\frac{t}{2}\right)\right|
\end{eqnarray}
denotes the Clausen function~\cite{L-1981}, and differentiating with respect to $\omega$ (justified, for $\beta=2$, by the decay rate of the Fourier coefficients in Eq.~(\ref{beyond-Dyson})), we obtain
\begin{eqnarray} \label{cos-F-reg} \fl \qquad
    \sum_{\ell=1}^{\infty} \ell \left( \delta I_\ell^{(2)} + \frac{1}{2 \pi^2 \ell^2}\right) \cos (\omega \ell)
    = \frac{d}{d\omega} \Bigg\{ {\rm Im\,} \int_{0}^{\infty} d\lambda \, \mathbb{E} \big[ e^{i\omega n(\lambda)} \big] \nonumber\\
 \qquad\qquad\qquad\qquad\qquad
        -\frac{1}{2} \cot\left( \frac{\omega}{2}\right)
        + \frac{1}{2 \pi^2} {\rm Cl}_2(\omega)
    \Bigg\}.
\end{eqnarray}
Taken at $\omega=0$, the left-hand side yields the definition of $C_2^{(1)}$. On the right-hand side, the derivative at $\omega=0$ can be evaluated by virtue of Lemma~\ref{lemma-integral}, whose Eq.~(\ref{lemma-integral-eq}) implies that the expression inside the braces in Eq.~(\ref{cos-F-reg}), after expanding two more contributions from $\cot(\omega/2)$ and ${\rm Cl}_2(\omega)$, admits the expansion
\begin{eqnarray} \label{sin-F-reg-small} \fl\qquad\qquad
     \left( \frac{1}{12} -\frac{1}{2\pi^2} \log (2\pi)\right) \omega + o(\omega),
     \qquad \omega \to 0.
\end{eqnarray}
Therefore, the right-hand side of Eq.~(\ref{cos-F-reg}) at $\omega=0$ equals
\begin{eqnarray}
    \frac{1}{12} -\frac{1}{2\pi^2} \log (2\pi).
\end{eqnarray}
This completes the proof of Proposition~\ref{main-P}.
\end{proof}

\subsection{Variance of the $L$-th ordered eigenvalue and proof of the ``$1/6$'' formula}\label{Duality-section}

The asymptotic expansion of the variance of the $L$-th ordered eigenvalue, as $L\rightarrow\infty$, at $\beta=2$ is provided by Proposition~\ref{main-L}.

\begin{proof}[Proof of Proposition~\ref{main-L}] Following the conventions introduced prior to Eq.~(\ref{dI-ell}), the $L$-th ordered eigenlevel can be expressed as a sum of $L$ consecutive spacings,
\begin{eqnarray}\label{L-th-ordered}
    \lambda_L = \sum_{\ell=1}^L s_\ell.
\end{eqnarray}
Consequently, its variance takes the form
\begin{eqnarray} \label{var-theta-L}
    {\rm var}_2[\lambda_L] = \sum_{|\ell|\le L-1} \left(L- |\ell|\right) \delta I_{|\ell|}^{(2)}
    = L \sigma_0^{(2)}(L) -\sigma_1^{(2)}(L),
\end{eqnarray}
where we have defined the quantities (for $k=0,1$):
\begin{eqnarray}\label{sigma-k}
    \sigma_k^{(2)}(L) = \sum_{|\ell|\le L-1} |\ell|^k \delta I_{|\ell|}^{(2)}.
\end{eqnarray}
To develop the asymptotic expansions of $\sigma_k^{(2)}(L)$, we treat the two cases separately.
\noindent\newline\newline
(i) For $\sigma_0^{(2)}(L)$, we make use of Pandey's sum rule, Eq.~(\ref{Pandey}), to write down
\begin{eqnarray} \label{sigma-0}
        \sigma_0^{(2)}(L) =  C_2^{(0)} - 2\sum_{\ell=L}^{\infty} \delta I_\ell^{(2)}, \quad C_2^{(0)}=0.
\end{eqnarray}
Substituting Eq.~(\ref{beyond-Dyson}) and applying the Euler-Maclaurin formula yields, as $L\rightarrow\infty$:
\begin{eqnarray} \label{L-sigma-0}\fl \quad
        L \sigma_0^{(2)}(L) =  \frac{1}{\pi^2} + \frac{1}{2\pi^2 L}
        + \frac{1}{\pi^4 L^2}
        \left(  \log(2\pi L) + \gamma + \frac{\pi^2-9}{6}
        \right) + o(L^{-2}).
\end{eqnarray}
\noindent\newline
(ii) To treat $\sigma_1^{(2)}(L)$, we first regularize the sum with the help of the expansion Eq.~(\ref{beyond-Dyson}). Simple algebra reduces this to
\begin{eqnarray} \label{sigma-1}\fl \qquad
    \sigma_1^{(2)}(L) = 2 C_2^{(1)}
                - \frac{1}{\pi^2} \sum_{\ell=1}^{L-1} \frac{1}{\ell}
                - 2 \sum_{\ell=L}^{\infty} \ell \left( \delta I_\ell^{(2)} + \frac{1}{2 \pi^2 \ell^2} \right),
\end{eqnarray}
where $C_2^{(1)}$ is provided by the sum rule in Proposition~\ref{main-P}. Applying the Euler-Maclaurin formula to the series above, we derive:
\begin{eqnarray} \label{sigma-1-asymp}\fl \qquad
    \sigma_1^{(2)}(L) =   -\frac{1}{\pi^2}  \log L  + 2 \left( C_2^{(1)}  - \frac{\gamma}{2 \pi^2} \right)
               +\frac{1}{2\pi^2 L}    \nonumber\\
               \qquad +  \frac{3}{2 \pi^4 L^2} \left( \log(2\pi L) +\gamma +\frac{\pi^2-24}{18}\right) + o(L^{-2}).
\end{eqnarray}
Combining Eqs.~(\ref{var-theta-L}), (\ref{L-sigma-0}), (\ref{sigma-1-asymp}) and (\ref{L-1-2}) completes the proof.
\end{proof}

\begin{proof}[Proof of Theorem~\ref{main-T}] 
To complete the proof, the large-$L$ expansion of the number variance in Eq.~(\ref{var-2}) is required. Using the asymptotic expansions~\cite{DLMF} of the cosine and sine integrals
defined in Eq.~(1.6),
\begin{eqnarray} \label{Ci-exp}
    {\rm Ci}(x)
    &=& \frac{\sin x}{x}
    -\frac{\cos x}{x^2}
    +O(x^{-3}), \\
    \label{Si-exp}
    {\rm Si}(x)
    &=& \frac{\pi}{2}
    -\frac{\cos x}{x}
    -\frac{\sin x}{x^2}
    +\frac{2\cos x}{x^3}
    +O(x^{-4}), 
\end{eqnarray}
as $x\rightarrow+\infty$, and setting $x=2\pi L$,
we obtain
\begin{eqnarray}\label{var-2-int-L-exp}
    \fl\qquad
    {\rm var}_2 [{\mathcal N}(L)]
    = \frac{1}{\pi^2}
    \Big(\log(2\pi L)+\gamma+1\Big)
    -\frac{\cos(2\pi L)}{4\pi^4 L^2}
    +o(L^{-2}).
\end{eqnarray}
For positive integer $L$, the oscillatory factor
$\cos(2\pi L)$ equals unity. Subtracting
Eq.~(\ref{var-2-lambda-L}) of Proposition~\ref{main-L}
from Eq.~(\ref{var-2-int-L-exp}) then gives
Eq.~(\ref{main-T-eq}), completing the proof.
\end{proof}

\subsection{Small-$\omega$ expansion of ${\rm Im\,}\mathbb{E} \big[ e^{i\omega n(\lambda)} \big]$ and related integral}
\label{technical-section}
The proof of Proposition~\ref{main-P} relies on the small-$\omega$ expansion of the integral in Eq.~(\ref{sin-F-reg}). This expansion is established in Lemma~\ref{lemma-integral} below, together with the auxiliary technical Lemma~\ref{L-g-Lemma}.

\begin{lemma}\label{lemma-integral}
Let $n(\lambda)$ be the counting function of the ${Sine}_2$ point process with unit mean local density. Then, as $\omega\to 0$, the following expansion holds:
\begin{eqnarray} \label{lemma-integral-eq}
    {\rm Im} \int_{0}^{\infty} d\lambda \,\mathbb{E} \big[ e^{i\omega n(\lambda)} \big] = \frac{1}{\omega}
    + \frac{\omega}{2\pi^2} \log\left( \frac{|\omega|}{2\pi} \right) - \frac{\omega}{2\pi^2} + o(\omega).
\end{eqnarray}
\end{lemma}

\begin{remark}
This complements our previous result for positive frequencies
\begin{eqnarray} \label{real-integral-eq}
    {\rm Re} \int_{0}^{\infty} d\lambda \,\mathbb{E} \big[ e^{i\omega n(\lambda)} \big] = \frac{\omega}{4\pi}
    + \frac{\omega^3}{8\pi^3} \log\left( \frac{\omega}{2\pi} \right) + O(\omega^4),
\end{eqnarray}
which was derived in the context of the small-$\omega$ expansion of the power spectrum of the circular unitary ensemble; see Proposition~4.10 of Ref.~\cite{RK-2023}.
\end{remark}

\begin{proof}[Proof of Lemma~\ref{lemma-integral}] Noting that the imaginary part of the integral in Eq.~(\ref{lemma-integral-eq}) is an odd function of $\omega$, we assume below that $\omega>0$. Throughout the proof, we adopt the notation and technique of Ref.~\cite{RK-2023}. In particular, we introduce $\tilde\omega = \omega/2\pi$ and the counting function $\tilde{n}(\lambda)$ corresponding to the ${\rm Sine}_\beta$ process with mean density $\varrho = 1/2\pi$. The cumulant expansion
\begin{eqnarray} \label{MGF-p1}
    \log \mathbb{E} \big[ e^{2i\pi \tilde\omega \tilde{n}(\lambda)} \big] = \sum_{\ell=1}^{\infty}  \frac{(2i\pi)^\ell}{\ell!} \tilde\omega^\ell \kappa_\ell (\lambda),
\end{eqnarray}
where $\kappa_\ell (\lambda) = \langle\!\langle \tilde{n}^\ell(\lambda) \rangle\!\rangle$ denotes the $\ell$-th cumulant of the counting function $\tilde{n}(\lambda)$, is our starting point. As $\tilde\omega\to 0$, we have
\begin{eqnarray} \label{MGF-p2}\fl\qquad\qquad
    {\rm Im\,}\mathbb{E} \big[ e^{2i\pi\tilde\omega \tilde{n}(\lambda)} \big] = \frac{\sin(\tilde{\omega} \lambda)}{\lambda^{2\tilde{\omega}^2}}
        + {\rm Im}\left(
        \frac{e^{i\tilde\omega\lambda}}{\lambda^{2\tilde\omega^2}}
        \sum_{k\ge 2} (2 i \pi\tilde\omega)^k p_k(\lambda)
        \right)\!,
\end{eqnarray}
where~\cite{Note-Cumulant-Polynomials}
\begin{eqnarray} \label{MGF-p3}
    p_k(\lambda) = \sum_{j=1}^{\lfloor k/2 \rfloor} \frac{1}{j!} \sum_{\substack{m_1+\cdots+m_j=k\\ m_\nu\ge 2}}
 \frac{d_{m_1}(\lambda)\cdots d_{m_j}(\lambda)}{m_1! \cdots m_j!}.
\end{eqnarray}
Here, $d_2(\lambda)=\tilde\kappa_2(\lambda)$, whereas $d_\ell(\lambda)=\kappa_\ell(\lambda)$ for $\ell\ge 3$. The oscillatory terms in Eq.~(\ref{MGF-p2}) originate from the mean
\begin{eqnarray}
    \kappa_1(\lambda) = {\mathbb E}[\tilde{n}(\lambda)] = \frac{\lambda}{2\pi}.
\end{eqnarray}
The power-law factor $\lambda^{-2\tilde{\omega}^2}$ is generated by the divergent part $\pi^{-2} \log\lambda$ of the number variance $\kappa_2(\lambda) = {\rm var}[\tilde{n}(\lambda)]$, while $\tilde{\kappa}_2(\lambda)$ denotes its regular part,
\begin{eqnarray} \label{kappa-2-regular}
    \tilde{\kappa}_2(\lambda) = \kappa_2(\lambda) - \frac{1}{\pi^2} \log \lambda.
\end{eqnarray}
For the explicit expression for $\kappa_2(\lambda) = {\rm var}_2 [{\mathcal N}(\lambda/2\pi)]$, we refer the reader to Eq.~(\ref{var-2}).

Integrating the expansion in Eq.~(\ref{MGF-p2}), we write
\begin{eqnarray}\label{Im-integral}\fl\qquad\qquad
    {\rm Im\,} \int_{0}^{\infty} d\lambda \, \mathbb{E}\big[ e^{2i\pi\tilde\omega \tilde{n}(\lambda)} \big] =
    J_0(\tilde\omega) + {\rm Im} \sum_{k\ge 2} (2i\pi)^k \tilde\omega^k J_k(\tilde\omega),
\end{eqnarray}
where
\begin{eqnarray} \label{J0-def}
    J_0(\tilde\omega) = \int_{0}^{\infty} d\lambda \,\frac{\sin(\tilde\omega \lambda)}{\lambda^{2\tilde\omega^2}},
\end{eqnarray}
and
\begin{eqnarray}\label{J-reminder}
    J_k(\tilde\omega) = \int_{0}^{\infty} d\lambda \,\frac{e^{i\tilde\omega \lambda}}{\lambda^{2 \tilde\omega^2}} p_k(\lambda).
\end{eqnarray}
We now analyze the small-$\tilde\omega$ behavior of these integrals.
\noindent\newline\newline
{\it The integral $J_0(\tilde\omega)$.}---The exact evaluation
\begin{eqnarray}\label{J0-exact}
    J_0(\tilde\omega) = \Gamma\left(
        1 - 2\tilde{\omega}^2
    \right) \cos\left(\pi \tilde\omega^2 \right)\, \tilde\omega^{2\tilde\omega^2-1},\quad 0 < \tilde\omega < \frac{1}{2},
\end{eqnarray}
yields the asymptotic expansion
\begin{eqnarray}\label{J0-asympt}
    J_0(\tilde\omega) = \frac{1}{\tilde\omega} + 2 \tilde{\omega} \log\tilde{\omega}+ 2\gamma  \tilde{\omega}  + O(\tilde{\omega}^3 \log^2\tilde\omega).
\end{eqnarray}
In the context of the proof of Proposition~\ref{main-P}, we retain terms up to and including $O(\tilde\omega)$.
\noindent\newline\newline
{\it Contributions of integrals $\tilde\omega^k J_k(\tilde\omega)$.}---For the forthcoming small-$\tilde\omega$ analysis, only soft estimates of the cumulants near $\lambda=0$ and as $\lambda\to\infty$ are needed. Near $\lambda=0$, we use the common bound
\begin{eqnarray}\label{d-m-0}
    d_m(\lambda) = O(|\log\lambda|), \qquad \lambda\to0^+, \quad m\ge2,
\end{eqnarray}
which in particular accommodates the logarithmic singularity of
$d_2(\lambda)=\tilde\kappa_2(\lambda)$. For $m\ge3$, the stronger estimate~\cite{RK-2023}
$d_m(\lambda)=O(\lambda)$ implies this bound.

For large $\lambda$, we retain the estimates
\begin{eqnarray}\label{d-m-infinity} \fl \qquad
    d_m(\lambda)=d_m(\infty)+O\left(\frac{\log^m \lambda}{\lambda}\right)\!,
    \quad
    d_m^\prime(\lambda)=O\left(\frac{\log^m \lambda}{\lambda^2}\right)\!,
    \quad \lambda\to\infty.
\end{eqnarray}
(Note that for odd values of $m$, the plateau term $d_m(\infty)$ vanishes; see Ref.~\cite{RK-2023}.) For $m=2,3$, Eq.~(\ref{d-m-infinity}) is a deliberately weakened version of the sharper asymptotics
\begin{eqnarray}\label{kappa-2-expansion}
    \tilde{\kappa}_2(\lambda) = \tilde{\kappa}_2(\infty) - \frac{\cos \lambda}{\pi^2 \lambda^2} + O(\lambda^{-3}), \quad
    \lambda \rightarrow \infty,
\end{eqnarray}
where
\begin{eqnarray}\label{kappa-2-plateau}
    \tilde{\kappa}_2(\infty) = \frac{1}{\pi^2} \left(
        1 + \gamma
    \right),
\end{eqnarray}
and~\cite{Note-Cumulant-Misprint}
\begin{eqnarray}\label{kappa-3-expansion}\fl\qquad
    \kappa_3(\lambda) = \frac{3}{\pi^3 \lambda} \left\{
    1 + 2 \frac{\sin \lambda}{\lambda}(\log \lambda+\gamma)
    \right\} + O\left( \frac{\log\lambda}{\lambda^3} \right)\!, \quad
    \lambda \rightarrow \infty,
\end{eqnarray}
respectively. This weakening is convenient because it allows all cumulants, including $\tilde\kappa_2(\lambda)$, to be treated on the same footing.

We now decompose
\begin{eqnarray}\label{p-k-decomposition}
    p_k(\lambda) = p_k^{\rm lin}(\lambda) + p_k^{\rm prod}(\lambda),
\end{eqnarray}
where
\begin{eqnarray}
    p_k^{\rm lin}(\lambda) = \frac{d_k(\lambda)}{k!}
\end{eqnarray}
is the term linear in $d_{(\cdot)}(\lambda)$, whereas $p_k^{\rm prod}(\lambda)$ collects the genuine product terms~\cite{Note-Product-Terms}, as given by Eq.~(\ref{MGF-p3}) with $j\ge 2$.

For the forthcoming asymptotic analysis,
we introduce the integral whose small-$\tilde\omega$ behavior
is established in Lemma~\ref{L-g-Lemma}:
\begin{equation}\label{L-g-int}
    L_g(\tilde\omega)
    = \int_0^\infty d\lambda\,
    \frac{e^{i\tilde\omega\lambda}}
         {\lambda^{a^2\tilde\omega^2}}\,g(\lambda),
\end{equation}
where $a>0$ is fixed and $\tilde\omega>0$
is sufficiently small. In the applications below, we set
$a=\sqrt{2}$.
\noindent\newline\newline
{\it (i) Terms linear in $d_{(\cdot)}(\lambda)$.} For every $k\ge 2$, Lemma~\ref{L-g-Lemma}
applies to $g=d_k$ in view of Eqs.~(\ref{d-m-0}) and (\ref{d-m-infinity}). Hence, as
$\tilde\omega\to0$,
\begin{eqnarray}\label{L-d-k}
L_{d_k}(\tilde\omega) = \int_0^\infty d\lambda \frac{e^{i\tilde\omega \lambda}}{\lambda^{2\tilde\omega^2}}\, d_k(\lambda)
= \frac{i\,d_k(\infty)}{\tilde\omega} + o(\tilde\omega^{-1}),
\end{eqnarray}
and therefore
\begin{eqnarray}\label{p-k-lin}
 \tilde\omega^k \int_{0}^{\infty} d\lambda \,\frac{e^{i\tilde\omega \lambda}}{\lambda^{2 \tilde\omega^2}} p_k^{\rm lin}(\lambda)
 = \frac{i\,d_k(\infty)}{k!}\,\tilde\omega^{k-1} + o(\tilde\omega^{k-1}).
\end{eqnarray}
In particular, for $k=2$ this yields
\begin{eqnarray} \label{p-2-lin}
    {\rm Im} \left[ (2i\pi)^2 \tilde\omega^2 \int_{0}^{\infty}
    d\lambda \frac{e^{i\tilde\omega \lambda}}{\lambda^{2\tilde\omega^2}} p_2^{\rm lin}(\lambda)\right] =
    -2(1+\gamma) \tilde\omega + o(\tilde\omega).
\end{eqnarray}

This is the only contribution beyond $J_0(\tilde\omega)$ in Eq.~(\ref{Im-integral}) that contributes nontrivially to the limit $\tilde\omega\to 0$ in the proof of Proposition~\ref{main-P}; see the discussion below Eq.~(\ref{cos-F-reg}). Indeed, for every $k\ge 3$, Eq.~(\ref{p-k-lin}) yields
\begin{eqnarray} \label{p-k-g-3-lin}
    (2i\pi)^k \tilde\omega^k \int_{0}^{\infty}
    d\lambda \frac{e^{i\tilde\omega \lambda}}{\lambda^{2\tilde\omega^2}} p_k^{\rm lin}(\lambda)
    =
    O(\tilde\omega^{k-1}) = o(\tilde\omega),
\end{eqnarray}
so all linear terms with $k\ge 3$ are negligible.
\noindent\newline\newline
{\it (ii) Product terms.} Fix a monomial
\begin{eqnarray} \label{g-multi-def}
    g_{m_1,\dots, m_j}(\lambda) = d_{m_1}(\lambda)\cdots d_{m_j}(\lambda),
\end{eqnarray}
where $m_1+\cdots + m_j=k$, $m_\nu\ge 2$, and $j\ge 2$. Equations~(\ref{d-m-0}) and (\ref{d-m-infinity}) imply the expansions
\begin{eqnarray} \label{g-multi-0}
    g_{m_1,\dots, m_j}(\lambda) = O(|\log\lambda|^k),\qquad \lambda\to0^+,
\end{eqnarray}
and
\begin{eqnarray}\fl\qquad \label{g-multi-infty}
    g_{m_1,\dots, m_j}(\lambda) = g_{m_1,\dots, m_j}(\infty) + O\left(\frac{\log^k \lambda}{\lambda}\right)\!\!,
    \quad g_{m_1,\dots, m_j}^\prime(\lambda) = O\left(\frac{\log^k \lambda}{\lambda^2}\right)\!\!
\end{eqnarray}
as $\lambda\to\infty$. Hence, Lemma~\ref{L-g-Lemma} applies to each such monomial, yielding
\begin{eqnarray}\fl\qquad
    L_{g_{m_1,\dots, m_j}}(\tilde\omega) = \int_0^\infty d\lambda \frac{e^{i\tilde\omega \lambda}}{\lambda^{2\tilde\omega^2}} \, g_{m_1,\dots, m_j}(\lambda) =
    \frac{i g_{m_1,\dots, m_j}(\infty)}{\tilde\omega} + o(\tilde\omega^{-1}).
\end{eqnarray}
Multiplying by $\tilde\omega^k$, we conclude that every product contribution of total weighted degree $k$ is $O(\tilde\omega^{k-1})$, so that
\begin{eqnarray}\label{p-k-prod}
    (2i\pi)^k \tilde\omega^k \int_{0}^{\infty}
    d\lambda \frac{e^{i\tilde\omega \lambda}}{\lambda^{2\tilde\omega^2}} p_k^{\rm prod}(\lambda)
    =
    O(\tilde\omega^{k-1}) = o(\tilde\omega).
\end{eqnarray}
Here, we have used the fact that $j\ge 2$ entails $k\ge 4$.
\noindent\newline\newline
{\it Combining all contributions.}---Equations~(\ref{Im-integral}), (\ref{J0-def}), (\ref{J-reminder}), (\ref{J0-asympt}), (\ref{p-k-decomposition}), (\ref{p-2-lin}), (\ref{p-k-g-3-lin}), and (\ref{p-k-prod}) yield
\begin{eqnarray}\label{Im-integral-final}    \fl\qquad
    {\rm Im\,} \int_{0}^{\infty} d\lambda \, \mathbb{E}\big[ e^{i \omega n(\lambda)} \big]
    =  \frac{1}{2\pi}
    {\rm Im\,} \int_{0}^{\infty} d\lambda \, \mathbb{E}\big[ e^{2i\pi\tilde\omega \tilde{n}(\lambda)} \big] \nonumber\\
    \fl\qquad\qquad\qquad\qquad\qquad = \frac{1}{2\pi}
    \left(\frac{1}{\tilde\omega} + 2\tilde\omega \log\tilde\omega - 2\tilde\omega\right) + o(\tilde\omega).
\end{eqnarray}
Finally, substituting $\tilde\omega=\omega/2\pi$ completes the proof.
\end{proof}

\begin{lemma}\label{L-g-Lemma}
Let $L_g(\tilde\omega)$ be defined by
Eq.~(\ref{L-g-int}), where $a>0$ is fixed and
$g(\lambda)$ is a differentiable function on $(0,\infty)$
that satisfies the following conditions:
(i) there exists $P\ge 0$ such that
$g(\lambda)=O(|\log\lambda|^P)$ as $\lambda\rightarrow0^+$,
and (ii) there exist $g_\infty\in{\mathbb R}$ and $M\ge0$
such that
\begin{eqnarray}\label{g-asymptotics}
    g(\lambda)
    = g_\infty
    +O\left(\frac{\log^M\lambda}{\lambda}\right)\!\!,
    \quad
    g^\prime(\lambda)
    =O\left(\frac{\log^M\lambda}{\lambda^2}\right)\!\!,
\end{eqnarray}
as $\lambda\rightarrow\infty$.
Then, as $\tilde\omega\rightarrow0^+$,
\begin{equation}\label{L-g-statement}
    L_g(\tilde\omega)
    =\frac{i g_\infty}{\tilde\omega}
    +o(\tilde\omega^{-1}).
\end{equation}
\end{lemma}

\begin{proof}
Split the integral $L_g(\tilde\omega)$ as
\begin{eqnarray}\label{L-g-split}
L_g(\tilde\omega)
    = g_\infty L_0(\tilde\omega) &+& \int_0^1 d\lambda
\frac{e^{i\tilde\omega \lambda}}{\lambda^{a^2\tilde\omega^2}} \left( g(\lambda)-g_\infty \right) \nonumber\\
&+&
\int_1^\infty d\lambda
\frac{e^{i\tilde\omega \lambda}}{\lambda^{a^2\tilde\omega^2}} \left( g(\lambda)-g_\infty \right)\!,
\end{eqnarray}
where
\begin{eqnarray}\label{L-0}\fl\qquad
L_0(\tilde\omega) = \int_0^{\infty} d\lambda \, \frac{e^{i\tilde\omega \lambda}}{\lambda^{a^2\tilde\omega^2}} =
i \, e^{-i\pi a^2\tilde\omega^2/2} \,\tilde\omega^{a^2\tilde\omega^2-1}\Gamma(1-a^2\tilde\omega^2), \quad a^2\tilde\omega^2<1,
\end{eqnarray}
and consider each term separately. As $\tilde\omega\rightarrow 0^+$, the first term yields
\begin{eqnarray}\label{L-0-expansion-0}
   g_\infty L_0(\tilde\omega) = \frac{i g_\infty}{\tilde\omega} + O(\tilde\omega \log\tilde\omega).
\end{eqnarray}
To treat the compact integral in Eq.~(\ref{L-g-split}), we note that, by assumption~(i),
one has $g(\lambda)-g_\infty = O(|\log \lambda|^P)$ as $\lambda\to0^+$. For $\tilde\omega$
sufficiently small, it holds that $\lambda^{-a^2\tilde\omega^2}\le \lambda^{-1/2}$, so the integrand is dominated by
an integrable function on $(0,1)$. Therefore,
\begin{eqnarray} \label{L-g-2nd}
\int_0^1 d\lambda\,
\frac{e^{i\tilde\omega \lambda}}{\lambda^{a^2\tilde\omega^2}}
\left( g(\lambda)-g_\infty \right)
= O(1),
\qquad \tilde\omega\to0^+.
\end{eqnarray}
To analyze the tail integral in Eq.~(\ref{L-g-split}), we integrate it by parts:
\begin{eqnarray}\label{L-g-3rd}
   \int_1^\infty d\lambda
    \frac{e^{i\tilde\omega \lambda}}{\lambda^{a^2\tilde\omega^2}} \left( g(\lambda)-g_\infty \right) &=& -\frac{e^{i\tilde\omega}}{i\tilde\omega}\left(
        g(1) - g_\infty
        \right) \nonumber\\
        &-& \frac{1}{i\tilde\omega} \int_{1}^{\infty} d\lambda \, e^{i\tilde\omega \lambda}
        \frac{d}{d\lambda} \left( \frac{g(\lambda) - g_\infty}{\lambda^{a^2\tilde\omega^2}}
        \right)\!\!.
\end{eqnarray}
By assumption~(\ref{g-asymptotics}), the integrand on the right-hand side of Eq.~(\ref{L-g-3rd}) is
\begin{eqnarray}
    O\left(
        \frac{\log^M\lambda}{\lambda^2}\right),\quad \lambda\rightarrow\infty
\end{eqnarray}
uniformly for $\tilde\omega$ small. Thus, after multiplying Eq.~(\ref{L-g-3rd}) by $\tilde\omega$, dominated convergence applies to the last integral. Using Eqs.~(\ref{L-g-split}), (\ref{L-0-expansion-0}), (\ref{L-g-2nd}), and (\ref{L-g-3rd}), we obtain
\begin{eqnarray}
    \lim_{\tilde\omega\rightarrow 0^+} \tilde\omega L_g(\tilde\omega) = i g_\infty + i \left( g(1)-g_\infty \right)
    + i \int_{1}^{\infty} d\lambda\, g^\prime(\lambda).
\end{eqnarray}
Cancellation of the last two terms yields
\begin{eqnarray} \label{L-g-pre-final}
    \lim_{\tilde\omega\rightarrow 0^+} \tilde\omega L_g(\tilde\omega) = i g_\infty.
\end{eqnarray}
This completes the proof, since Eq.~(\ref{L-g-pre-final}) is equivalent to Eq.~(\ref{L-g-statement}).
\end{proof}
\noindent
{\textbf{Acknowledgements.}}---The authors thank F.~Bornemann for providing us with the MATLAB package for numerical evaluation of Fredholm determinants and helpful correspondence. This work was supported by the Israel Science Foundation through the Grant No. 956/24.

\newpage

\section*{References}
\fancyhead{} \fancyhead[RE,LO]{References}
\fancyhead[LE,RO]{\thepage}


\begin{thebibliography}{99}
\bibitem{Note-Product-Terms} The product terms first appear at weighted degree $k=4$, since $p_k^{\rm prod}=0$ for $k=2,3$.

\smallskip
\end{thebibliography}
\end{document}